\documentclass[11pt, letterpaper]{article}

\usepackage{algorithm}
\usepackage{algcompatible}
\usepackage{algpseudocode}
\usepackage{graphicx}
\usepackage{amssymb}
\usepackage{amsthm}
\usepackage{amsmath}
\usepackage{bbm}
\usepackage{tikz}
\usepackage{tikz-qtree}
\usepackage[htt]{hyphenat}
\usepackage{authblk}

\usepackage{thmtools} 
\usepackage{thm-restate}

\usepackage[margin=1in, bmargin=1.15in, tmargin=1.1in]{geometry}
\usepackage[bookmarks=false]{hyperref}
\usepackage{cleveref}
\usepackage{cite}
\theoremstyle{definition}
\newtheorem{theorem}{Theorem}
\newtheorem{lemma}[theorem]{Lemma}

\newtheorem{cor}[theorem]{Corollary}

\newtheorem*{remark}{Remark}

\usepackage{xspace}
\newcommand\eps{\varepsilon}
\newcommand\PP{\xspace \Pr}
\newcommand\EE{\xspace \mathbb{E}}
\newcommand\E{\EE}
\newcommand \midd{\; \middle| \;}
\begin{document}

\title{The Multiplicative Version of Azuma's Inequality, with an Application to Contention Analysis}

\author{William Kuszmaul}
\author{Qi Qi}
\affil{Massachusetts Institute of Technology\\\texttt {\{kuszmaul, qqi\}@mit.edu}}

\date{}

\maketitle

\thispagestyle{empty}

\begin{abstract}
  Azuma's inequality is a tool for proving concentration bounds on
  random variables. The inequality can be thought of as a natural
  generalization of additive Chernoff bounds. On the other hand, the
  analogous generalization of multiplicative Chernoff bounds does not appear to be widely known.

  We formulate a multiplicative-error version of Azuma's
  inequality. We then show how to apply this new inequality in order
  to greatly simplify (and correct) the analysis of contention delays
  in multithreaded systems managed by randomized work stealing.
\end{abstract}

\newpage
\pagenumbering{arabic}

\section{Introduction}

One of the most widely used tools in algorithm analysis is the Chernoff bound, which gives a concentration inequality on sums of independent random variables. The Chernoff bound exists in many forms, but the two most common variants are the additive and multiplicative bounds:

\begin{theorem}[Additive Chernoff Bound]
Let $X_1, \ldots, X_n \in \{0, 1\}$ be independent random variables, and let $X=\sum_{i=1}^nX_i$. Then for any $\eps>0$, 
\[\PP\left[X\geq \E[X] +\eps\right]\leq \exp\left(-\frac{2\eps^2}{n}\right).\]
\end{theorem}

\begin{theorem}[Multiplicative Chernoff Bound]
Let $X_1, \ldots, X_n \in \{0, 1\}$ be independent random variables. Let $X=\sum_{i=1}^nX_i$ and let $\mu = \E[X]$. Then for any $\delta>0$,
\[\PP[X\geq(1+\delta)\mu]\leq \exp\left(-\frac{\delta^2\mu}{2+\delta}\right)\]
and for any $0<\delta<1$,
\[\PP[X\leq(1-\delta)\mu]\leq \exp\left(-\frac{\delta^2\mu}{2}\right).\]
\end{theorem}

Although the additive Chernoff bound is often convenient to use, the multiplicative bound can in some cases be much stronger. Suppose, for example, that $X_1, X_2, \ldots, X_n$ each take value $1$ with probability $(\log n) / n$. By the additive bound, one can conclude that $\sum_i X_i = O(\sqrt{n \log n})$ with high probability in $n$. On the other hand, the multiplicative bound can be used to show that $\sum_i X_i = O(\log n)$ with high probability in $n$. In general, whenever $\E[X] \ll n$, the multiplicative bound is more powerful.

\paragraph{Handling dependencies with Azuma's inequality.}
Chernoff bounds require that the random variables $X_1, X_2, \ldots, X_n$ be independent. In many algorithmic applications, however, the $X_i$'s are not independent, such as when analyzing algorithms in which $X_1, X_2, \ldots, X_n$ are the results of decisions made by an \emph{adaptive} adversary over time. When analyzing these applications (see, e.g., \cite{gyorfi1999simple, costello2011randomized,azar2013loss,aspnes2011randomized,gopinathan2011strategyproof,dinur2003revealing,yaroslavtsev2019approximate,komm2014randomized,levi2014local,chakrabarty2016widetilde, bollobas2003directed,bender2020contention,bender2019achieving,auletta2000randomized,kumar2001approximation}), a stronger inequality known as \emph{Azuma's inequality} is often useful.

\begin{theorem}[Azuma's inequality]
Let $Z_0, Z_1, \ldots, Z_n$ be a supermartingale, meaning that $\EE[Z_i\mid Z_0, \ldots, Z_{i-1}]\leq Z_{i-1}$. Assume additionally that $|Z_i-Z_{i-1}|\leq c_i$. Then for any $\eps>0$,
\[\PP[Z_n-Z_0\geq\eps]\leq \exp\left(-\frac{\eps^2}{2\sum_{i=1}^nc_i^2}\right).\]
\label{thm:azuma}
\end{theorem}

By applying Azuma's inequality to the exposure martingale for a sum
$\sum_i X_i$ of random variables, one arrives at the following
corollary, which is often useful in analyzing randomized algorithms (for direct applications of Corollary \ref{cor:azuma}, see, e.g., \cite{komm2014randomized,casteigts2019design,levi2014local,chakrabarty2016widetilde,cesa2004generalization,duchi2011dual}).

\begin{cor}
  Let $X_1, X_2, \ldots, X_n$ be random variables satisfying
  $X_i \in [0, c_i]$. Suppose that
  $\E[X_i \mid X_1, \ldots, X_{i - 1}] \le p_i$ for all $i$. Then for any $\eps>0$,
  \[\PP\left[\sum_i X_i \geq \sum_i p_i + \eps\right] \leq \exp\left(-\frac{\eps^2}{2\sum_{i=1}^nc_i^2}\right).\]

  \label{cor:azuma}
\end{cor}

\paragraph{This paper: an inequality with multiplicative error.} Azuma's inequality can be viewed as a generaliation of \emph{additive}
Chernoff bounds. In this paper, we formulate the multiplicative
analog to Azuma's inequality. As we will discuss later, this result can also be viewed as a corollary of other known results \cite{freedman1975tail, new1, new2, new3}.  

\begin{restatable}{theorem}{MAINTHM}
Let $Z_0, Z_1, \ldots, Z_n$ be a supermartingale, meaning that $\EE[Z_i\mid Z_0, \ldots, Z_{i-1}]\leq Z_{i-1}$. Assume additionally that $-a_i\leq Z_i-Z_{i-1}\leq b_i$, where $a_i+b_i=c$ for some constant $c>0$ independent of $i$. Let $\mu=\sum_{i=1}^na_i$. Then for any $\delta >0$, 
\[\PP[Z_n-Z_0\geq\delta\mu]\leq \exp\left(-\frac{\delta^2\mu}{(2+\delta)c}\right).\]
\label{thm:main}
\end{restatable}

This theorem yields the following corollary.
\begin{restatable}{cor}{MAINCOR}
Let $X_1, \ldots, X_n \in [0, c]$ be real-valued random variables with $c>0$. Suppose that $\E[X_i \mid X_1, \ldots, X_{i - 1}] \le a_i$ for all $i$. Let  $\mu=\sum_{i=1}^na_i$. Then for any $\delta >0$, 
\[\PP\left[\sum_i X_i\geq (1 + \delta)\mu \right]\leq \exp\left(-\frac{\delta^2\mu}{(2+\delta)c}\right).\]
\label{cor:main}
\end{restatable}

In the same way that multiplicative Chernoff bounds are in some cases much stronger than additive Chernoff bounds, the multiplicative Azuma's inequality is in some cases much stronger than the standard (additive) Azuma's inequality, as occurs, in particular, when $\sum_i a_i \ll c n$.

Our work is targeted
towards algorithm designers. Our hope is that Theorem \ref{thm:main}
will simplify the task of analyzing randomized algorithms, providing
an instrument that can be used in place of custom Chernoff bounds and
ad-hoc combinatorial arguments.

\paragraph{Extensions.}
We present two extensions of Theorem \ref{thm:main} and Corollary
\ref{cor:main}.

In Section \ref{sec:app2}, we generalize Theorem \ref{thm:main} so
that $a_1, a_2, \ldots, a_n$ are determined by an \emph{adaptive
  adversary}. This means that each $a_i$ can be partially a function
of $Z_0, \ldots, Z_{i - 1}$. As long as the $a_i$'s are restricted to
satisfy $\sum_i a_i \le \mu$, then the bound from Theorem
\ref{thm:main} continues to hold. We also discuss several applications
of the adaptive version of the theorem.

In Appendix \ref{sec:app}, we extend Theorem \ref{thm:main} to give a
lower tail bound. In particular, just as Theorem
\ref{thm:main} gives an upper tail bound for supermartingales, a
similar approach gives a lower tail bound for submartingales (also
with multiplicative error).

\paragraph{An application: work stealing. } 
In order to demonstrate the power of Theorem \ref{thm:main} we revisit
a classic result in multithreaded scheduling. In the problem of
scheduling multithreaded computations on parallel computers, a
fundamental question is how to decide when one processor should
``steal'' computational threads from another. In the seminal paper,
\emph{Scheduling Multithreaded Computations by Work Stealing}
\cite{BlumofeLe99}, Blumofe and Leiserson presented the first provably
good work-stealing scheduler for multithreaded computations with
dependencies. The paper has been influential to both theory and
practice, amassing almost two thousand citations, and inspiring the
Cilk Programming language and runtime system \cite{BlumofeJoKuLe96,
  FrigoLeRa98, SchardlLeLe18}.

One result in \cite{BlumofeLe99} is an analysis of the so-called
$(P, M)$-recycling game\footnote{Not to be confused with the ball recycling
game of \cite{BenderChCoFa19}.}, which models the contention incurred by a
randomized work-stealing algorithm. By combining the analysis of the
$(P, M)$-recycling game with a delay-sequence argument, the authors
are able to bound the execution time and communication cost of their
randomized work-stealing algorithm.

The $(P, M)$-recycling game takes place on $P$ bins which are
initially empty. In each step of the game, if there are $k$ balls
presently in the bins, then the player selects some value
$j \in \{0, 1, \ldots, P - k\}$ and then tosses $j$ balls at random
into bins. At the end of each step, one ball is removed from each
non-empty bin. The game continues until $M$ total tosses have been made. The goal of the player is to
maximize the \emph{total delay} experienced by the balls, where the
delay of a ball $b$ thrown into a bin $i$ is defined to be the number
of balls already present in bin $i$ at the time of $b$'s throw. Lemma
6 of \cite{BlumofeLe99} states that, even if the player is an adaptive
adversary, the total delay is guaranteed to be at most
$O(M + P \log P + P \log \epsilon^{-1})$ with probability at least
$1 - \epsilon$.

In part due to lack of good analytical tools, the authors of
\cite{BlumofeLe99} attempt to analyze the $(P, M)$-recycling game via
a combinatorial argument. Unfortunately, the argument fails to notice
certain subtle (but important) dependencies between random variables,
and consequently the analysis is incorrect.\footnote{We thank Charles Leiserson of MIT, one of the original authors of \cite{BlumofeLe99}, for
  suggesting that the analysis in \cite{BlumofeLe99} should be
  revisited, which led us to write this paper.}

In Section \ref{sec:application}, we give a simple and short analysis
of the $(P, M)$-recycling using Theorem \ref{thm:main}. We also
explain why the same argument does not follow from the standard
Azuma's inequality. In addition to being simpler (and more correct)
than the analysis in \cite{BlumofeLe99}, our analysis enables the
slightly stronger bound of $O(M + P \log \epsilon^{-1})$.

\subsection{Related Work}
Although Chernoff bounds are often attributed to Herman Chernoff, they
were originally formulated by Herman Rubin (see discussion in
\cite{Chernoff14}). Azuma's inequality, on the other hand, was
independently formulated by several different authors, including
Kazuoki Azuma \cite{Azuma67}, Wassily Hoeffding \cite{Hoeffding94},
and Sergei Bernstein \cite{Bernstein1} (although in a slightly
different form). As a result, the inequality is sometimes also
referred to as the Azuma-Hoeffding inequality.

The key technique used to prove Azuma's inequality is to apply
Markov's inequality to the moment generating function of a random
variable. This technique is well understood and has served as the
foundation for much of the work on concentration inequalities in
statistics and probability theory \cite{Deviations1, Deviations2,
  Deviations3, Deviations4, Deviations5, Deviations6, Deviations7,
  Deviations8, Deviations9, Deviations10, Deviations11, Deviations12,
  Deviations13, Deviations14, McDiarmid89} (see \cite{DeviationsBook}
or \cite{chung2006concentration} for a survey).  Extensive work has
been devoted to generalizing Azuma's inequality in various ways. For
example, Bernstein-type inequalities parameterize the concentration
bound by the $k$th moments of the random variables being summed
\cite{Bernstein1, Bernstein2, Deviations1, Deviations2, Deviations3,
  Deviations4, Deviations5, Deviations6, Deviations7,
  Deviations8}. Most of the research in this direction has been
targeted towards applications in statistics and probability theory,
rather than to theoretical computer science.

The main contribution of this paper is to explicitly formulate the multiplicative analogue of Azuma's inequality, and to discuss its application within algorithm analysis. We emphasise that the proof of the inequality is not, in itself, a substantial contribution, since the inequality is relatively straightforward to derive by combining the proof of the multiplicative Chernoff bound with that of Azuma's inequality. Nonetheless, by presenting the theorem as a tool that can be directly referenced by algorithm designers, we hope to simplify the task of proving concentration bounds within the context of algorithm analysis. 

Besides Azuma's inequality, there are several other generalizations of
Chernoff bounds that are used in algorithm analysis. Chernoff-style
bounds have been shown to apply to sums of random variables that are
negatively associated, rather than independent, and several works have
developed useful techniques for identifying when random variables are
negatively associated \cite{NA1, NA2, NA3, NA4}. Another common
approach is to show that a sum $X$ of not necessarily independent random variables
is stochastically dominated by a sum $X'$ of independent random
variables (see Lemma 3 of \cite{Dominance}), thereby allowing for the
application of Chernoff bounds to $X$.

\vspace{.3 cm}

\noindent\textbf{Addendum: } Since writing the original version of this paper, we have learned of several other references whose results imply (either implicitly or explicitly) variations of the concentration bounds in this paper. Freedman (in Proposition 2.1 of \cite{freedman1975tail}) gives a concentration bound in terms of the variance of each martingale difference (this can be viewed as a variation of Bernstein's inequality \cite{Bernstein1, Bernstein2, Deviations1, Deviations2, Deviations3,
  Deviations4, Deviations5, Deviations6, Deviations7,
  Deviations8}). Using the fact that, for $[0, 1]$-random variable $X$, we have $\operatorname{Var}(X) \le \E[X]$, one can recover from this an adaptive version of Azuma's inequality as a corollary. Several other sources also prove results that have the multiplicative Azuma's inequality as direct corollaries; see, e.g., Lemma 10 of \cite{new1}, Theorem 2.2 of \cite{new2}, and the results in \cite{new3}.

\section{Multiplicative Azuma's Inequality}

In this section we prove the following theorem and corollary. 

\MAINTHM*

\MAINCOR*

\bigskip

\noindent We start our proof by establishing a simple inequality.

\begin{lemma}
For any $t>0$ and any random variable $X$ such that $\EE[X]\leq0$ and $-a\leq X\leq b$, 
\[\EE\left[e^{tX}\right]\leq\exp\left(\frac{a}{a+b}\left(e^{t(a+b)}-1\right)-ta\right).\]
\label{lem:ineq}
\end{lemma}

\begin{proof}
Consider the linear function $f$ defined on $[-a, b]$ that passes through points $(-a, e^{-ta})$ and $(b, e^{tb})$. Since $e^{tx}$ is convex, Jensen's inequality states that $f$ upper bounds $e^{tx}$, implying that $\EE[e^{tX}]\leq\EE[f(X)]$. Since $f$ is linear, $\EE[f(X)]$ only depends on $\EE[X]$, and one can derive that 
\[\EE[f(X)]=\frac{b-\EE[X]}{a+b}e^{-ta}+\frac{a+\EE[X]}{a+b}e^{tb}.\]
This quantity is maximized when $\EE[X]$ is maximized at $\EE[X]=0$. Therefore,
\begin{flalign*}
\EE[e^{tX}]&\leq\frac{b}{a+b}e^{-ta}+\frac{a}{a+b}e^{tb}\\
&=e^{-ta}\left(1+\frac{a}{a+b}\left(e^{t(a+b)}-1\right)\right)\\
&\leq\exp\left(\frac{a}{a+b}\left(e^{t(a+b)}-1\right)-ta\right).
\end{flalign*}
\end{proof}

\begin{proof}[Proof of Theorem \ref{thm:main}]
By Markov's inequality, for any $t>0$ and $v$,
\begin{flalign}
\PP[Z_n-Z_0\geq v]&=\PP\left[e^{t(Z_n-Z_0)}\geq e^{tv}\right] \nonumber \\
&\leq\frac{\EE[e^{t(Z_n-Z_0)}]}{e^{tv}}. \label{eq:Markov}
\end{flalign}
Let $X_i=Z_i-Z_{i-1}$. Since $Z_i$ is a supermartingale, for any $i$, $\EE[X_i\mid Z_0, \ldots, Z_{i-1}]\leq 0$. Moreover, from the assumptions in the problem, $-a_i\leq X_i\leq b_i$. Therefore, Lemma \ref{lem:ineq} applies to $X=(X_i\mid Z_0, \ldots,  Z_{i-1})$, and we have
\begin{equation}
\EE[e^{tX_i}\mid Z_0, \ldots, Z_{i-1}]\leq\exp\left(\frac{a_i}{c}\left(e^{tc}-1\right)-ta_i\right). \label{eq:ineqcor}
\end{equation}
In the following derivation, which will involve expectations of expectations, it will be important to understand which random variables each expectation is taken over. We will adopt the notation $\E_S[f(S)]$ to denote an expectation taken over a set of random variables $S$. Using \eqref{eq:ineqcor} along with the law of total expectation, which states that $\EE_{X, Y}[A]=\EE_X[\EE_Y[A|X]]$ for any random variable $A$ that is a function of random variables $X$ and $Y$, we derive
\begin{flalign*}
\EE_{Z_0,X_1,\ldots,X_{i-1},X_i}\left[\prod_{j=1}^ie^{tX_j}\right]
&=\EE_{Z_0,X_1,\ldots,X_{i-1}}\left[\EE_{X_i}\left[\prod_{j=1}^ie^{tX_j}\midd Z_0,X_1,\ldots,X_{i-1}\right]\right]\\
&=\EE_{Z_0,X_1,\ldots X_{i-1}}\left[\left(\prod_{j=1}^{i-1}e^{tX_j}\right)\EE_{X_i}\left[e^{X_i}\midd Z_0, X_1 \ldots, X_{i-1}\right]\right]\\
&=\EE_{Z_0,X_1,\ldots X_{i-1}}\left[\left(\prod_{j=1}^{i-1}e^{tX_j}\right)\EE_{X_i}\left[e^{X_i}\midd Z_0, Z_1, \ldots, Z_{i-1}\right]\right]\\
&\leq\EE_{Z_0,X_1,\ldots X_{i-1}}\left[\left(\prod_{j=1}^{i-1}e^{tX_j}\right)\exp\left(\frac{a_i}{c}\left(e^{tc}-1\right)-ta_i\right)\right]\\
&=\exp\left(\frac{a_i}{c}\left(e^{tc}-1\right)-ta_i\right)\EE_{Z_0,X_1,\ldots X_{i-1}}\left[\prod_{j=1}^{i-1}e^{tX_j}\right].
\end{flalign*}
By applying the above inequality iteratively, we arrive at the following:
\begin{flalign*}
\EE[e^{t(Z_n-Z_0)}]&=\EE_{Z_0, X_1, \ldots, X_n}\left[\prod_{i = 1}^n e^{t X_i} \right]\\
&\leq\prod_{i = 1}^n \exp\left(\frac{a_i}{c}(e^{tc}-1)-ta_i\right)\\
&=\exp\left(\frac\mu c\left(e^{tc}-1\right)-t\mu\right).
\end{flalign*}
By \eqref{eq:Markov}, we have
\[\PP[Z_n-Z_0\geq v]\leq\exp\left(\frac{\mu}{c}\left(e^{tc}-1\right)-t\mu-tv\right).\]
Plugging in $t=(\ln(1+\delta))/c$ and $v=\delta\mu$ for $\delta>0$ yields
\begin{flalign*}
\PP[Z_n-Z_0\geq\delta\mu]\leq&\exp\left(\frac{\mu\delta}{c}-\frac{\mu}{c}\ln(1+\delta)-\frac{\mu}{c}\delta\ln(1+\delta)\right)\\
=&\exp\left(\frac{\mu}{c}\left(\delta-(1+\delta)\ln(1+\delta)\right)\right).
\end{flalign*}
For any $\delta>0$,
\[\delta-(1+\delta)\ln(1+\delta)\leq-\frac{\delta^2}{2+\delta}, \]
which can be seen by inspecting the derivative of both sides.\footnote{Consider $f(x)=x/(1+x)-\ln(1+x)+x^2/((1+x)(2+x))$. Then $f(0)=0$ and $f'(x)=-x^2/\left((1+x)(2+x)^2\right)\leq 0$ for $x\geq 0$. Therefore, $f(x)\leq 0$ for $x\geq 0$, and the inequality holds for $\delta>0$.}
As a result,
\[\PP[Z_n-Z_0\geq\delta\mu]\leq\exp\left(-\frac{\delta^2\mu}{(2+\delta)c}\right).\]

\end{proof}

\begin{remark}
A stronger but more unwieldy bound may sometimes be helpful. By skipping the approximation of $\delta-(1+\delta)\ln(1+\delta)$, we derive
\[\PP[Z_n-Z_0\geq\delta\mu]\leq\left(\frac{e^\delta}{(1+\delta)^{(1+\delta)}}\right)^{\mu/c}.\]
\end{remark}

We conclude the section by proving Corollary \ref{cor:main}.

\begin{proof}[Proof of Corollary \ref{cor:main}]
Define $Z_i=\sum_{j=1}^i(X_j-a_j)$. Note that $Z_i-Z_{i-1}=X_i-a_i$. The given condition 
\[\EE[X_i\mid X_1,\ldots, X_{i-1}]\leq a_i\]
implies that
\[\EE[Z_i-Z_{i-1}\mid Z_0,\ldots, Z_{i-1}]=\EE[Z_i-Z_{i-1}\mid X_1,\ldots, X_{i-1}]\leq 0\]
and thus that $Z_i$ is a supermartingale. Moreover, as each $X_i\in[0, c]$, we have that $Z_i-Z_{i-1}\geq-a_i$, $Z_i-Z_{i-1}\geq c-a_i$. Setting $\mu=\sum_{i=1}^na_i$, Theorem \ref{thm:main} implies
\[\PP[Z_n-Z_0\geq\delta\mu]\leq\exp\left(-\frac{\delta^2\mu}{(2+\delta)c}\right).\]
We may break down $Z_n-Z_0$ as
\begin{flalign*}
Z_n-Z_0&=\sum_{i=1}^n(Z_i-Z_{i-1})\\
&=\sum_{i=1}^n(X_i-a_i)\\
&=\sum_{i=1}^nX_i-\mu.
\end{flalign*}
Therefore,
\[\PP\left[\sum_{i=1}^nX_i\geq(1+\delta)\mu\right]\leq\exp\left(-\frac{\delta^2\mu}{(2+\delta)c}\right).\]
\end{proof}

\section{Analyzing the $(P, M)$-Recyling Game}\label{sec:application}

In this section we revisit the analysis of the $(P, M)$-recycling game given in \cite{BlumofeLe99}. We begin by defining the game and explaining why the analysis given in \cite{BlumofeLe99} is incorrect. Then we apply Theorem \ref{thm:main} to obtain a simple and correct analysis.

\subsection{Defining the Game}

The $(P, M)$-recycling game is a combinatorial game, in which balls labelled $1$ to $P$ are tossed at random into $P$ bins. Initially, all $P$ balls are in a reservoir separate from the $P$ bins. At each step of the game, the player executes the following two operations in sequence: 

\begin{enumerate}
\item The player chooses some of the balls in the reservoir (possibly all and possibly none). For each of these balls, the player removes it from the reservoir, selects one of the $P$ bins uniformly and independently at random, and tosses the ball into it.

\item The player inspects each of the $P$ bins in turn, and for each bin that contains at least one ball, the player removes any one of the balls in the bin and returns it to the reservoir. 
\end{enumerate}
The player is permitted to make a total of $M$ ball tosses. The game ends when $M$ ball tosses have been made and all balls have been removed from the bins and placed back in the reservoir. The player is allowed to base their strategy (how many/which balls to toss) depending on outcomes from previous turns. 

After each step $t$ of the game, there are some number $n_t$ of balls left in the bins. The \emph{total delay} is defined as $D = \sum_{t=1}^T n_t$, where $T$ is the total number of steps in the game. Equivalently, if we define the \emph{delay} of a ball $b$ being tossed into a bin $i$ to be the number of balls already present in bin $i$ at the time of the toss, then the total delay is the sum of the delays of all ball tosses.

We would like to give high probability bounds on the total delay, no matter what strategy the player takes.

\subsection{An Incorrect Analysis of the Recycling Game}

The following bound is given by \cite{BlumofeLe99}.
\begin{theorem}[Lemma 6 in \cite{BlumofeLe99}]
  For any $\eps>0$, with probability at least $1-\eps$, the total delay in the $(P, M)$-recycling game is $O(M+P\log P+P\log\eps^{-1})$.
  \label{thm:recyclingold}
\end{theorem}

In order to prove Theorem \ref{thm:recyclingold}, the authors \cite{BlumofeLe99} sketch a complicated combinatorial analysis of the game. Define the indicator random variable $x_{ir}$ to be $1$ if the $i$th toss of ball $1$ is delayed by ball $r$, and $0$ otherwise. A key component in the analysis \cite{BlumofeLe99} is to show that\footnote{In fact, the analysis requires a somewhat stronger property to be shown. But for simplicity of exposition, we focus on this simpler variant.}, for any set $R \subseteq [P]$ of balls,
\begin{equation}
\Pr[x_{ir} \text{ for all } r \in R] \le P^{-|R|}.
\label{eq:recyclingfalse}
\end{equation}

Unfortunately, due to subtle dependencies between the random variables $x_{ir}$, \eqref{eq:recyclingfalse} is not true (or even close to true). To see why, suppose that the player (i.e., the adversary) takes the following strategy: Throw balls $A = \{2, 3, \ldots, P\}$ in the first step. If the balls in $A$ do not land in the same bin, then wait $P - 1$ steps until all bins are empty, and throw the balls in $A$ again. Continue rethrowing until there is some step $t$ in which all of the balls in $A$ land in the same bin. At the end of step $t$, remove ball $2$, leaving balls $3, 4, \ldots, P$ in the same bin as each other. Then on step $t + 1$ perform the first throw of ball $1$.

If $M$ is sufficiently large so that all balls in $A$ almost certainly land together before the process ends, then the probability that the first throw of ball $1$ lands in the same bin as balls $3, 4, \ldots, P$ is approximately $1/{P}$. In contrast, \eqref{eq:recyclingfalse} claims to bound the same probability by  $1 / P^{P - 2}$.

The difficulty of proving Theorem \ref{thm:recyclingold} via an ad-hoc combinatorial argument is further demonstrated by another error in \cite{BlumofeLe99}'s analysis. Throughout the proof, the authors define $m_i$ to be the number of times that ball $i$ is thrown, and then treat each $m_i$ as taking a fixed value. In actuality, however, the $m_i$'s are random variables that are partially controlled by an adversary (i.e., the player of the game), meaning that the outcomes of the $m_i$'s may be linked to the outcomes of the $x_{i, r}$'s. This consideration adds even further dependencies that must be considered in order to obtain a correct analysis.

\subsection{A Simple and Correct Analysis Using Multiplicative Azuma's Inequality}

We now give a simple (and correct) analysis of the $(P, M)$-recycling game using the multiplicative version of Azuma's inequality. In fact, we prove a slightly stronger bound than Theorem \ref{thm:recyclingold}.

\begin{theorem}
  For any $\eps>0$, with probability at least $1-\eps$, the total delay in the $(P, M)$-recycling game is $O(M+P\log(1/\eps))$.
  \label{thm:recyclingnew}
\end{theorem}
\begin{proof}
For $i=1, 2, \ldots, M$, define the delay $X_i$ of the $i$th toss to be the number of balls in the bin that the $i$th toss lands in, not counting the $i$th toss itself. The total delay can be expressed as $D=\sum_{i=1}^MX_i$. 

As the player's strategy can adapt to the outcomes of previous tosses,
the $X_i$'s may have complicated dependencies. Nonetheless, since
there are at most $P - 1$ balls present at time of the $i$th toss, we
know that $X_i \in [0, P]$. Moreover, since the toss selects a bin
$\{1, 2, \ldots, P\}$ at random, each ball present at the time of the
toss has probability $1 / P$ of contributing to the delay $X_i$. Thus,
no matter the outcomes of $X_1, \ldots, X_{i - 1}$, we have that
$\E[X_i \mid X_1, \ldots, X_{i - 1}] \le (P - 1)/{P} \le 1$. We
can therefore apply Corollary \ref{cor:main}, with $a_i = 1$ for all
$i$ and $c = P$, to deduce that
\begin{equation}
  \Pr[D \ge (1+\delta) M] \le \exp\left(-\frac{\delta^2M}{(2+\delta)P}\right).
  \label{eq:cor}
\end{equation}

If $M\geq P\ln(1/\eps)$, we may substitute $\delta=2$ into \eqref{eq:cor} to derive
$\PP[D\geq 3M]\leq\exp\left(-M / P\right)\leq\eps$.
If $M\leq P\ln(1/\eps)$, we may instead substitute $\delta=2P\ln(1/\eps)/M$. As $\delta\geq2$, we have $\delta/(2+\delta)\geq 1/2$, and we derive
$\PP\left[D\geq M+2P\ln (1/\eps) \right]\leq\exp\left(- \delta M / (2P) \right)=\eps$.

In either case, $\PP\left[D\geq 3M+2P\ln (1/\eps)\right]\leq\eps$,
which proves the theorem statement.
\end{proof}

\subsection{Why Standard Azuma's Inequality Does Not Suffice}

In order to fully understand the proof of Theorem
\ref{thm:recyclingnew}, it is informative to consider what happens if
we attempt to use (the standard) Azuma's inequality to analyze
$D = \sum_{i = 1}^M X_i$. Applying Corollary \ref{cor:azuma} with
$c_i = P$ for all $i$, we get that
\begin{equation}
    \Pr[D > (1+\delta)M] \le \exp\left(- \frac{(\delta M)^2}{2 M P^2} \right) = \exp\left(- \frac{\delta^2 M}{2P^2} \right) .
    \label{eq:badbound}
\end{equation}
In contrast, for $\delta \ge 2$, Corollary \ref{cor:main} gives a bound of
\begin{equation}
    \Pr[D > (1+\delta)M] \le \exp\left(- \frac{\delta^2 M}{(2 + \delta)P} \right) \le \exp\left(- \frac{\delta M}{2P} \right).
    \label{eq:goodbound}
\end{equation}
Since $D \le PM$ trivially, the interesting values for $\delta$ are $\delta \le P$. On the other hand, for all $\delta$ satisfying $2 \le \delta < P$, the bound given by \eqref{eq:goodbound} is stronger than the bound given by \eqref{eq:badbound}.
The reason that the multiplicative version of Azuma's does better than the additive version is that the random variables $X_i$ have quite small means, meaning that the $a_i$'s used by the multiplicative bound are much smaller than the $c_i$'s used by the additive bound. When $\delta$ is a constant, this results in a full factor-of-$\Theta(P)$ difference in the exponent achieved by the two bounds. It is not possible to derive a $O(M+P\log\eps^{-1})$ high probability bound with \eqref{eq:badbound} alone.

\section{Adversarial Multiplicative Azuma's Inequality}\label{sec:app2}

In this section, we extend Theorem \ref{thm:main} and Corollary \ref{cor:main} in order to allow for the values $a_1, a_2, \ldots, a_n$ and $b_1, b_2, \ldots, b_n$ to be random variables that are determined adaptively. Formally, we define the supermartingale $Z_0, \ldots, Z_n$ with respect to a filtration, and then defining $a_1, a_2, \ldots, a_n$ and $b_1, b_2, \ldots, b_n$ to be predictable processes with respect to that same filtration.

The statement of Theorem \ref{thm:adv} uses several notions that are standard in probability theory (see, e.g., \cite{Roch15} and \cite{Billingsley08} for formal definitions) but less standard in theoretical computer science. 

\begin{theorem}
Let $Z_0, \ldots Z_n$ be a supermartingale with respect to the filtration $F_0, \ldots, F_{n}$, and let $A_1, \ldots, A_n$ and $B_1, \ldots, B_n$ be predictable processes with respect to the same filtration. Suppose there exist values $c>0$ and $\mu$, satisfying that $-A_i\leq Z_i-Z_{i-1}\leq B_i$, $A_i+B_i=c$, and $\sum_{i=1}^nA_i\leq\mu$ (almost surely). Then for any $\delta>0$,
\[\PP[Z_n-Z_0\geq\delta\mu]\leq\exp\left(-\frac{\delta^2\mu}{(2+\delta)c}\right).\]
\label{thm:adv}
\end{theorem}

\begin{cor}
  Suppose that Alice constructs a sequence of random variables $X_1, \ldots X_n$, with $X_i\in[0, c], c>0$, using the following iterative process. Once the outcomes of $X_1, \ldots, X_{i - 1}$ are determined, Alice then selects the probability distribution $\mathcal{D}_i$ from which $X_i$ will be drawn; $X_i$ is then drawn from distribution $\mathcal{D}_i$. Alice is an adaptive adversary in that she can adapt $\mathcal{D}_i$ to the outcomes of $X_1, \ldots, X_{i - 1}$. The only constraint on Alice is that $\sum_i \E[X_i \mid \mathcal{D}_i] \le \mu$, that is, the sum of the means of the probability distributions $\mathcal{D}_1, \ldots, \mathcal{D}_n$ must be at most $\mu$.

If $X = \sum_i X_i$, then for any $\delta>0$,
\[\PP[X\geq(1+\delta)\mu]\leq\exp\left(-\frac{\delta^2\mu}{(2+\delta)c}\right).\]
\label{cor:adv}
\end{cor}


\begin{remark} Formally, a filtration $F_0, \ldots, F_{n - 1}$ is a sequence of $\sigma$-algebras such that $F_i \subseteq F_{i + 1}$ for each $i$. Informally, one can simply think of the $F_i$'s as revealing ``random bits''. For each $i$, $F_{i}$ reveals the set of random bits used to determine all of $Z_0, \ldots, Z_i$, $A_0, \ldots, A_i$, and $B_0, \ldots, B_i$. The fact that $Z_0, Z_1, \ldots, Z_n$ is a martingale with respect to $F_0, F_1, \ldots, F_{n - 1}$ means simply that the random bits $F_i$ determine $Z_i$ (that is, $Z_i$ is \emph{$F_i$-measurable}), and that $\E[Z_i \mid F_{i -1}] = Z_{i - 1}$. The fact that $A_1, \ldots, A_n$ and $B_1, \ldots, B_n$ are predictable processes, means simply that each $A_i$ and $B_i$ is determined by the random bits $F_{i-1}$ (that is, $A_i, B_i$ are \emph{$F_{i  - 1}$-measurable}).
\end{remark}
  
To prove Theorem \ref{thm:adv}, we prove the following key lemma:

\begin{lemma}
Let $Z_0, \ldots Z_n$ be a supermartingale with respect to the filtration $F_0, \ldots, F_{n}$, and let $A_1, \ldots, A_n$ and $B_1, \ldots, B_n$ be predictable processes with respect to the same filtration. Suppose there exist values $c>0$ and $\mu$, satisfying that $-A_i\leq Z_i-Z_{i-1}\leq B_i$, $A_i+B_i=c$, and $\sum_{i=1}^nA_i\leq\mu$ (almost surely). Then for any $t>0$,
\[\EE\left[e^{t(Z_n-Z_0)}\midd F_0\right]\leq\exp\left(\frac{\mu}{c}(e^{tc}-1)-t\mu\right).\]
\label{lem:keyadv}
\end{lemma}

\begin{proof}
We proceed by induction on $n$.

\noindent\textbf{The base case.} For $n=0$, $Z_n-Z_0=0$, and for any $c, t>0$, $(e^{tc}-1)/c-t>0$. Therefore, $\mu(e^{tc}-1)/c-t\mu\geq0=t(Z_n-Z_0)$, and the inequality holds.

\noindent\textbf{The inductive step.} Assume that this statement is true for $n-1$, and we shall prove it for $n$. 

The law of total expectation states that for any random variable $X$ and any $\sigma$-algebras $H_1\subseteq H_2$, $\EE[\EE[X\mid H_2]\mid H_1]=\EE[X\mid H_1]$. As $\{F_i\}$ is a filtration, we know $F_{i-1}\subseteq F_i$, and thus
\[\EE\left[e^{t(Z_n-Z_0)}\midd F_0\right]=\EE\left[\EE\left[e^{t(Z_n-Z_0)}\midd F_1\right]\midd F_0\right].\]

Since $e^{t(Z_1-Z_0)}$ is $F_1$-measurable, we can pull it out of the expectation as follows:
\begin{equation}
  \EE\left[\EE\left[e^{t(Z_n-Z_0)}\midd F_1\right]\midd F_0\right]=\EE\left[e^{t(Z_1-Z_0)}\cdot\EE\left[e^{t(Z_n-Z_1)}\midd F_1\right]\midd F_0\right]. \label{eq:F0outofF1}
\end{equation}

Let $Z'_i=Z_{i+1}$, $F'_i=F_{i+1}$, $A'_i=A_{i+1}$, $B'_i=B_{i+1}$. We know that $Z'_0,\ldots, Z'_{n-1}$ is a supermartingale with respect to $F'_0,\ldots,F'_{n-1}$. Additionally, we know $A'_1,\ldots,A'_{n-1}$ and $B'_1,\ldots,B'_{n-1}$ are predictable processes with respect to $F'_0,\ldots,F'_{n-1}$ satisfying that $-A'_i\leq Z'_i-Z'_{i-1}\leq B'_i$, $A'_i+B'_i=c$, and $\sum_{i=1}^{n-1}A'_i\leq\mu-(A_1\mid F_0)$. Therefore, we may apply our inductive hypothesis to derive
\begin{flalign}
\EE\left[e^{t(Z_n-Z_1)}\midd F_1\right] 
&=\EE\left[e^{t(Z'_{n-1}-Z'_0)}\midd F'_0\right] \nonumber \\
&\leq\left(\exp\left(\frac{\mu-A_1}{c}(e^{tc}-1)-t(\mu-A_1)\right)\midd F_0\right). \label{eq:F1ind}
\end{flalign}

Combining \eqref{eq:F0outofF1} and \eqref{eq:F1ind}, we find that
\[\EE\left[e^{t(Z_n-Z_0)}\midd F_0\right]=\EE\left[e^{t(Z_1-Z_0)}\cdot\exp\left(\frac{\mu-A_1}{c}(e^{tc}-1)-t(\mu-A_1)\right)\midd F_0\right].\]
As $A_1$ is $F_0$-measurable, we can pull the exponential term out of the expectation to arrive at
\begin{equation}
  \EE\left[e^{t(Z_n-Z_0)}\midd F_0\right]=\left(\exp\left(\frac{\mu-A_1}{c}(e^{tc}-1)-t(\mu-A_1)\right)\midd F_0\right)\EE\left[e^{t(Z_1-Z_0)}\midd F_0\right].
  \label{eq:missingA}
\end{equation}

Since $Z_i$ is a supermartingale, $\EE[Z_1-Z_0\mid F_0]\leq 0$. Therefore, Lemma \ref{lem:ineq} applies to $X=(Z_1-Z_0\mid F_0)$, $a=(A_1\mid F_0)$, $b=(B_1\mid F_0)$, and we have
\begin{equation}
  \EE[e^{t(Z_1-Z_0)}\mid F_0]\leq\left(\exp\left(\frac{A_1}{c}(e^{tc}-1)-tA_1\right)\midd F_0\right).
  \label{eq:hereAis}
\end{equation}

Combining \eqref{eq:missingA} and \eqref{eq:hereAis}, we have
\[\EE\left[e^{t(Z_n-Z_0)}\midd F_0\right]\leq\left(\exp\left(\frac{\mu}{c}(e^{tc}-1)-t\mu\right)\midd F_0\right)=\exp\left(\frac{\mu}{c}(e^{tc}-1)-t\mu\right).\]

\end{proof}

\begin{proof}[Proof of Theorem \ref{thm:adv}]
By Lemma \ref{lem:keyadv} and the law of total expectation, 
\begin{flalign*}
\EE[e^{t(Z_n-Z_0)}]
&=\EE[\EE[e^{t(Z_n-Z_0)}\mid F_0]]\\
&\leq\EE\left[\exp\left(\frac{\mu}{c}(e^{tc}-1)-t\mu\right)\right]\\
&=\exp\left(\frac{\mu}{c}(e^{tc}-1)-t\mu\right).
\end{flalign*}

The rest of the proof is identical to the proof of Theorem \ref{thm:main}.
\end{proof}

Corollary \ref{cor:adv} is a straightforward application of Theorem \ref{thm:adv}.
\begin{proof}[Proof of Corollary \ref{cor:adv}] Define the filtration $F_0, F_1, \ldots, F_{n }$ by
  $$F_i = \sigma(X_1, X_2, \ldots, X_i, \mathcal{D}_1, \mathcal{D}_2, \ldots,
  \mathcal{D}_{i+1}).$$ That is, $F_i$ is the smallest $\sigma$-algebra with respect to which all of $X_1, X_2, \ldots, X_{i}$ and $\mathcal{D}_1, \mathcal{D}_2, \ldots, \mathcal{D}_{i+1}$ are measurable.

Define $A_i = \E[X_i \mid \mathcal{D}_i]$ to be the expected value of $X_i$ once its distribution is determined, and $B_i = c - A_i$. Define $Z_0, \ldots, Z_n$ to be given by
  $$Z_i = \sum_{j = 1}^i X_i - \sum_{j = 1}^i A_i.$$

Since $A_i$ and $B_i$ are $D_i$-measurable and $F_{i-1}$ contains $D_i$, we know that $A_i$ and $B_i$ are also $F_{i-1}$-measurable, implying that they are predictable processes with respect to filtration $F_0, ..., F_n$.

As each $X_i$ is drawn from distribution $D_i$ after all of $X_1, ..., X_{i-1}$ and $D_1, ..., D_i$ have been determined, we have $\EE[X_i\mid F_{i-1}]=\EE[X_i\mid D_i]$. We can then compute that
\begin{flalign*}
\EE[Z_i\mid F_{i-1}]&=\EE[X_i-A_i+Z_{i-1}\mid F_{i-1}]\\
&=\EE[X_i\mid F_{i-1}]-A_i+Z_{i-1}\\
&=\EE[X_i\mid D_i]-A_i+Z_{i-1}\\
&=Z_{i-1},
\end{flalign*}
implying that $Z_0, ..., Z_n$ is a martingale with respect to filtration $F_0, ..., F_n$.
  
Finally, $\{Z_i\}$, $\{A_i\}$ and $\{B_i\}$ satisfy the requirements of Theorem \ref{thm:adv}, namely that $-A_i \leq Z_i - Z_{i - 1} \leq B_i$, that $A_i + B_i \ = c$, and that $\sum_i A_i \le \mu$. Thus, by Theorem \ref{thm:adv},
  \[\PP[Z_n \geq\delta\mu]\leq\exp\left(-\frac{\delta^2\mu}{(2+\delta)c}\right).\] Expanding out $Z_n$ gives
   \[\PP\left[\sum_i X_i \geq\sum_i A_i + \delta\mu\right]\leq\exp\left(-\frac{\delta^2\mu}{(2+\delta)c}\right),\] and thus we have
   \[\PP\left[\sum_i X_i \geq (1 + \delta)\mu\right]\leq\exp\left(-\frac{\delta^2\mu}{(2+\delta)c}\right),\] as desired.
 \end{proof}

\subsection{Applications of Theorem \ref{thm:adv} in Concurrent Work}

By allowing for an adaptive adversary, Theorem \ref{thm:adv} naturally
lends itself to applications with online adversaries. We conclude the
section by briefly discussing two applications of Theorem
\ref{thm:adv} that have arisen in several of our recent concurrent
works. In both cases, Theorem \ref{thm:adv} significantly simplified
the task of analyzing an algorithm.

\paragraph{Edge orientation in incremental forests} In \cite{edgeapplication}, Bender et al. consider the problem of edge orientation in an incremental forest. In this problem, edges $e_1, e_2, \ldots, e_k$ of a forest arrive one by one, and we are responsible for maintaining an orientation of the edges (i.e., an assignment of directions to the edges) such that every vertex has out-degree at most $O(1)$. As each edge $e_i$ arrives, we may need to flip the orientations of other edges in order to accommodate the newly arrived edge. The goal in \cite{edgeapplication} is to flip at most $O(\log \log n)$ orientations per edge insertion (with high probability). We refer to an edge insertion as a step.

A key component of the algorithm in \cite{edgeapplication} is that vertices may ``volunteer'' to have their out-degree incremented during a given step. During each step $i$, there are $\operatorname{polylog} n$ vertices $S_i$ that are eligible to volunteer, and each of these vertices volunteers with probability $1 / \operatorname{polylog} n$. The algorithm is designed to satisfy the property that each vertex $v$ can appear in at most $O(\log n)$ $S_i$'s.

An essential piece of the analysis is to show that, for any set $S$ of size $\operatorname{polylog} n$, the number of vertices in $S$ that ever volunteer is at most $|S| / 2$ (with high probability). On a given step $i$, the expected number of vertices in $S$ that volunteer is $|S \cap S_i| / \operatorname{polylog} n$. $S_i$ is partially a function of the algorithm's past randomness, and thus $S_i$ are effectively selected by an adaptive adversary, subject to the constraint that each vertex $v$ appears in at most $O(\log n)$ $S_i$'s.  By applying Theorem \ref{thm:adv}, one can deduce that the number of vertices in $S$ that volunteer is small (with high probability).

Note that, since $|S| = \operatorname{polylog} n$, a bound with additive error would not suffice here. Such a bound would allow for the number of vertices that volunteer to deviate by $\Omega(\sqrt{n})$ from its mean, which is larger than $|S|/2$.

\paragraph{Task scheduling against an adaptive adversary}
Another concurrent work to ours \cite{cupapplication} considers a
scheduling problem in which the arrival of new work to be scheduled is
controlled by a (mostly) adaptive adversary. In particular, although
the amount of new work that arrives during each step is fixed (to
$1 - \epsilon$), the tasks to which that new work is assigned are
determined by the adversary. The scheduling algorithm is then allowed
to select a single task to perform $1$ unit of work on. The goal is to
design a scheduling algorithm that prevents the backlog (i.e., the
maximum amount of unfinished work for any task) from becoming too
large.

Due to the complexity of the algorithm in \cite{cupapplication}, we
cannot explain in detail the application of Theorem \ref{thm:adv}. The
basic idea, however, is that the adversary must decide how to allocate
its resources across tasks over time, but that the adversary can adapt
(in an online fashion) to events that it has observed in the
past. Theorem \ref{thm:adv} allows for the authors of
\cite{cupapplication} to obtain Chernoff-style bounds on the number of
a certain ``bad events'' that occur, while handling the adaptiveness
of the adversary.

\section{Acknowledgments}

We thank Charles Leiserson of MIT for suggesting that the proof in \cite{BlumofeLe99} might warrant revisiting, and for giving useful feedback on the manuscript. We also thank Tao B. Schardl of MIT for several helpful conversations, and thank Kevin Yang of UC Berkeley for helpful discussions on probability theory. Finally, we would like to thank Chandra Chekuri and Kent Quanrud for helping identify several important pieces of related work.

\bibliographystyle{plain}
\bibliography{main}

\appendix

\section{Multiplicative Lower Tail Bounds}\label{sec:app}

In this section we prove a lower tail bound with multiplicative error for both the normal and the adversarial setting. Whereas Theorem \ref{thm:main} and Theorem \ref{thm:adv} allow us to bound the probability of a random variable substantially exceeding its mean, Theorem \ref{thm:lower} and Theorem \ref{thm:advlower} allow us to bound the probability of a random variable taking a substantially smaller value than its mean.
 
\begin{theorem}
Let $Z_0, Z_1, \ldots, Z_n$ be a submartingale, meaning that $\EE[Z_i\mid Z_0, \ldots, Z_{i-1}]\geq Z_{i-1}$. Assume additionally that $-a_i\leq Z_i-Z_{i-1}\leq b_i$, where $a_i+b_i=c$ for some constant $c>0$ independent of $i$. Let $\mu=\sum_{i=1}^na_i$. Then for any $0\leq\delta<1$, 
\[\PP[Z_n-Z_0\leq-\delta\mu]\leq \exp\left(-\frac{\delta^2\mu}{2c}\right).\]
\label{thm:lower}
\end{theorem}

\begin{cor}
Let $X_1, \ldots, X_n \in [0, c]$ be real-valued random variables with $c>0$. Suppose $\E[X_i \mid X_1, \ldots, X_{i - 1}] \geq a_i$ for all $i$. Let  $\mu=\sum_{i=1}^na_i$. Then for any $0\leq\delta<1$, 
\[\PP\left[\sum_i X_i\leq (1 - \delta)\mu \right]\leq \exp\left(-\frac{\delta^2\mu}{2c}\right).\]
\label{cor:lower}
\end{cor}

\begin{theorem}
Let $Z_0, \ldots Z_n$ be a submartingale with respect to the filtration $F_0, \ldots, F_{n}$, and let $A_1, \ldots, A_n$ and $B_1, \ldots, B_n$ be predictable processes with respect to the same filtration. Suppose there exist values $c>0$ and $\mu$, satisfying that $-A_i\leq Z_i-Z_{i-1}\leq B_i$, $A_i+B_i=c$, and $\sum_{i=1}^nA_i\geq\mu$ (almost surely). Then for any $\delta>0$,
\[\PP[Z_n-Z_0\leq-\delta\mu]\leq\exp\left(-\frac{\delta^2\mu}{2c}\right).\]
\label{thm:advlower}
\end{theorem}

\begin{cor}
  Suppose that Alice constructs a sequence of random variables $X_1, \ldots X_n$, with $X_i\in[0, c], c>0$, using the following iterative process. Once the outcomes of $X_1, \ldots, X_{i - 1}$ are determined, Alice then selects the probability distribution $\mathcal{D}_i$ from which $X_i$ will be drawn; $X_i$ is then drawn from distribution $\mathcal{D}_i$. Alice is an adaptive adversary in that she can adapt $\mathcal{D}_i$ to the outcomes of $X_1, \ldots, X_{i - 1}$. The only constraint on Alice is that $\sum_i \E[X_i \mid \mathcal{D}_i] \ge \mu$, that is, the sum of the means of the probability distributions $\mathcal{D}_1, \ldots, \mathcal{D}_n$ must be at least $\mu$.

If $X = \sum_i X_i$, then for any $\delta>0$,
\[\PP[X\leq(1-\delta)\mu]\leq\exp\left(-\frac{\delta^2\mu}{2c}\right).\]
\label{cor:advlower}
\end{cor}

\noindent We begin by proving Theorem \ref{thm:lower}. The proof is similar to the proof for the upper tail bound, with a different approximation used.

\begin{lemma}
For any $t<0$ and any random variable $X$ such that $\EE[X]\geq 0$ and $-a\leq X\leq b$,
\[\EE[e^{tX}]\leq\exp\left(\frac{a}{a+b}\left(e^{t(a+b)}-1\right)-ta\right).\]
\label{lem:ineq2}
\end{lemma}

\begin{proof}
Same as Lemma \ref{lem:ineq}.
\end{proof}

\begin{proof}[Proof of Theorem \ref{thm:lower}]
By Markov's inequality, for any $t<0$ and $v$,
\begin{flalign*}
\PP[Z_n-Z_0\leq v]&=\PP[t(Z_n-Z_0)\geq tv]\\
&=\PP\left[e^{t(Z_n-Z_0)}\geq e^{tv}\right]\\
&\leq\frac{\EE[e^{t(Z_n-Z_0)}]}{e^{tv}}.
\end{flalign*}
Let $X_i=Z_i-Z_{i-1}$. Since $Z_i$ is a submartingale, for any $i$, $\EE[X_i\mid Z_0, \ldots, Z_{i-1}]\geq 0$. Moreover, from the assumptions in the problem, $-a_i\leq X_i\leq b_i$. Therefore, Lemma \ref{lem:ineq2} applies to $X=(X_i\mid Z_0, \ldots,  Z_{i-1})$, and we have
\[\EE[e^{tX_i}\mid Z_0, \ldots, Z_{i-1}]\leq\exp\left(\frac{a_i}{c}\left(e^{tc}-1\right)-ta_i\right),\]
for any $t < 0$.
Using the same derivation as in the proof for Theorem \ref{thm:main}, we have
\[\PP[Z_n-Z_0\leq v]\leq\exp\left(\frac{\mu}{c}\left(e^{tc}-1\right)-t\mu-tv\right).\]
Plugging in $t={\ln(1-\delta)}/{c}$ and $v=-\delta\mu$ for $\delta>0$ yields
\[\PP[Z_n-Z_0\leq-\delta\mu]\leq\exp\left(\frac{\mu}{c}(-\delta-(1-\delta)\ln(1-\delta))\right).\]
For any $0\leq\delta<1$,
\[-\delta-(1-\delta)\ln(1-\delta)\leq-\frac{\delta^2}{2},\]
which can be seen by inspecting the derivative of both sides.\footnote{Consider $f(x)=-x/(1-x)-\ln(1-x)+x^2/(2(1-x))$. Then $f(0)=0$, and $f'(x)=-x^2/(2(1-x)^2)\leq0$ for $0\leq x<1$. Therefore, $f(x)\leq 0$ for $0\leq x<1$, and the inequality holds for $0\leq\delta<1$.}
As a result,
\[\PP[Z_n-Z_0\leq-\delta\mu]\leq\exp\left(-\frac{\delta^2\mu}{2c}\right).\]

\end{proof}

\begin{remark}
As with the upper tail bound, we may derive a stronger but more unwieldy bound of
\[\PP[Z_n-Z_0\leq-\delta\mu]\leq\left(\frac{e^{-\delta}}{(1-\delta)^{(1-\delta)}}\right)^{{\mu}/{c}}.\]
\end{remark}

The proof of Corollary \ref{cor:lower} is identical to the proof of Corollary \ref{cor:main}.

The proof of Theorem \ref{thm:advlower} can be obtained by combining the proofs of Theorem \ref{thm:adv} and Theorem \ref{thm:lower}.

The proof of Corollary \ref{cor:advlower} is identical to the proof of Corollary \ref{cor:adv}.

\end{document}